\documentclass[12pt]{article}

\usepackage{algorithm}
\usepackage{algpseudocode}
\usepackage{amsfonts}
\usepackage{amsmath}
\usepackage{amssymb}
\usepackage{amsthm}
\usepackage{booktabs}
\usepackage[footnotesize,bf]{caption}
\usepackage{enumitem}
\usepackage{float}
\usepackage{graphicx}
\usepackage[letterpaper, left=1.5in, right=1.5in, top=1.5in, bottom=1.5in]{geometry}
\usepackage[colorlinks, citecolor=blue, linkcolor=blue, urlcolor=blue]{hyperref}
\usepackage{multirow}
\usepackage{natbib}
\usepackage{pdfpages}
\usepackage{threeparttable}
\usepackage[rm,bf]{titlesec}
\usepackage{setspace}
\usepackage{url}

\newtheorem{theorem}{Theorem}

\newtheorem{hypothesis}[theorem]{Hypothesis}

\newtheorem{proposition}{Proposition}

\newcommand{\DI}{$2\times 2-$I}
\newcommand{\DE}{$2\times 2-$E}
\newcommand{\II}{$3\times3-$I}
\newcommand{\IE}{$3\times3-$E}

\begin{document}
\bibliographystyle{elsart-harv}
\title{Lies, Labels, and Mechanisms\footnote{Thanks to Samson Alva, David Cooper and seminar participants at the 2021 Global ESA meetings and TETC24 for useful comments. We thank Valon Vitaku for conducting experiments. All errors are our own.}}
\date{\today}

\author{Alexander L. Brown,\thanks{Department of
Economics, Texas A\&M University, College Station, TX 77843;
 \href{mailto:alexbrown@tamu.edu}{alexbrown@tamu.edu}; \href{http://people.tamu.edu/\%7Ealexbrown}{http://people.tamu.edu/$\sim$alexbrown}} \ \
 Ethan Park,\thanks{Berkeley Research Group, LLC, Houston, TX 77002;
 \href{mailto:ethanpark1@gmail.com}{ethanpark1@gmail.com}}\ \  and
 Rodrigo A. Velez\thanks{Department of
Economics, The University of Texas at San Antonio, San Antonio, TX 78249;
\href{mailto:rodrigo.velez@utsa.edu}{rodrigo.velez@utsa.edu}; \href{https://sites.google.com/site/rodrigoavelezswebpage/home}{https://sites.google.com/site/rodrigoavelezswebpage/home}}
 }
\maketitle

\begin{abstract}
\begin{singlespace}
We test whether lying aversion can steer equilibrium selection in mechanism design. In a principal–worker environment, the direct mechanism admits two dominant‑strategy equilibria: the designer’s target and a worker‑optimal outcome. We show this limitation persists for all robust mechanisms, then ask whether framing misreports as explicit lies helps. We develop a $2\times2$ experiment that varies direct vs.\ extended mechanisms with implicit vs.\ explicit messages. We find that framing misreporting of type as an explicit lie shifts play away from the worker‑optimal outcome toward truthful reporting, raising designer payoffs with minimal efficiency loss. These findings indicate that lying aversion is an effective lever for aligning behavior with social objectives.
\end{singlespace}
\medskip
\begin{singlespace}

\medskip


\end{singlespace}
\end{abstract}

\pagebreak
\section{Introduction}\label{Sec:Intro}

When a mechanism admits multiple equilibria, the designer’s intended outcome may fail to prevail. A theoretical literature shows that even a modest cost of dishonesty can resolve this problem \citep{MATSUSHIMA-2008-JET,DUTTA-SEN-2012-GEB,KARTIK-etal-2007-JET,KARTIK-etal-2014-GEB}. Empirically, however, it remains open whether lying aversion can be leveraged \emph{in practice} when agents’ interests diverge from the designer’s. We address this question in a principal–worker adverse-selection problem.

In our environment, the principal’s direct-revelation mechanism admits two dominant-strategy equilibria: (i) truthful play, which implements the principal’s social choice function, and (ii) a worker-optimal outcome sustained by non-cooperative coordination. We show that this limitation is generic: any mechanism that robustly implements—\`a la \citet{Bergemann-Morris-2005-Eca}—the principal’s objective also robustly implements the worker-optimal outcome (Proposition~\ref{Prop:Repullo}). Thus, robust implementation alone leaves room for undesirable equilibria.

\citet{Repullo-1985-RES} formulates the limitation in an abstract social-choice environment with no explicit economic interpretation. We map that structure into a natural principal–worker setting: tasks are \emph{delicate} or \emph{perfunctory}, and salaries are \emph{high} or \emph{low}. We reproduce the required payoff structure—experts value variety; beginners trade off delicate work against compensation. This translation turns a purely abstract impossibility into a testable, operational design problem and enables empirical evaluation.

We then examine whether lying aversion can mitigate the problem (see Sec.~\ref{Sec-behav-robust} for a review of the literature on lying aversion and its behavioral implications). Our experiment uses a $2\times2$ between-subjects design varying along two dimensions. First, we vary whether misreporting is framed as an \emph{explicit lie} by using type-labeled messages (explicit) versus neutral labels (implicit). Second, we vary whether the mechanism includes a \emph{leave the question blank} option that removes weak dominance from the deceptive action. The central hypotheses are: (i) the worker-optimal equilibrium is observed less often when the associated actions are framed as lies; and (ii) it is observed less often when those actions are no longer in dominant strategies.

The findings deliver a clear message. Framing the misreporting of type as an explicit lie shifts play away from the worker‑optimal equilibrium and toward the truthful one, with corresponding gains in designer payoffs. Relying only on dominance yields smaller effects. Taken together, the results show that lying aversion is a practical design lever. Simple message frames that make misreporting an explicit lie can materially improve mechanism performance even when agents have incentives to coordinate against the designer.

\section{Lying Aversion as a Design Lever}\label{Section-problem}

\subsection{Hiring talent with adverse selection}
\label{sec:RepulloGame}

A principal assigns tasks and compensation to a team of two workers. Each worker is either a beginner $(B)$ or an expert $(E)$. A worker’s expertise is private information, known only to that worker and not to the principal. For now, we leave unspecified what each worker knows about the other’s expertise. 

There are two tasks to be completed: a delicate task $(D)$ and a perfunctory task $(P)$. These tasks can be split between the two workers. Let $M$ denote the mixed task in which a worker performs half of $D$ and half of $P$. The principal can pay either a high salary $(H)$ or a low salary $(L)$. Each worker is assigned a task--salary pair; for example, $(H,D)$ denotes a high salary for performing the entire delicate task.

The principal’s objective is to pay a high salary to experts and a low salary to beginners, assign the delicate task to experts whenever at least one is available, and divide tasks among workers of equal skill. These objectives define a social choice function (SCF).

Workers’ preferences over assignments depend on their expertise. An expert’s payoff function, $\pi_i(\cdot \mid \textrm{E})$, satisfies
\begin{equation}\label{relexp}
\pi_i(\textrm{H,M}\mid \textrm{E}) 
> \pi_i(\textrm{H,D}\mid \textrm{E}) 
> \pi_i(\textrm{L,M}\mid \textrm{E}) 
> \pi_i(\textrm{L,P}\mid \textrm{E}).
\end{equation}
A beginner’s payoff function, $\pi_i(\cdot \mid \textrm{B})$, satisfies
\begin{equation}\label{relbeg}
\pi_i(\textrm{L,P}\mid \textrm{B}) = \pi_i(\textrm{H,M}\mid \textrm{B}) 
> \pi_i(\textrm{L,M}\mid \textrm{B}) = \pi_i(\textrm{H,D}\mid \textrm{B}).
\end{equation}
Ceteris paribus, a beginner prefers a higher salary. Performing the delicate task entails a cost. The cost of performing half of the delicate task rather than half of the perfunctory task is exactly offset by a high salary instead of a low salary.

\subsection{Standard robust implementation}

The principal seeks to design a contract that implements her objective and performs well under a range of informational assumptions. Following Harsanyi, we represent an informational assumption by a common prior over the type space $\Theta = \Theta_1 \times \Theta_2$, with $p \in \Delta(\Theta)$, where $\Delta(\Theta)$ denotes the set of probability measures on $\Theta$. Since all type and message spaces in our analysis are finite, probability distributions lie in a simplex. Anticipating strategic behavior by the workers, the principal requires the contract to be incentive compatible for all possible common priors.

It is well known that implementation of the principal’s desiderata based on simultaneous reports is possible only through mechanisms that admit ex post equilibria whose outcomes coincide with the principal’s objectives \citep{Bergemann-Morris-2005-Eca}. Moreover, since each worker’s payoff depends only on their own type, the associated social choice function (SCF) must be strategy-proof. That is, when used as the outcome function of a direct revelation mechanism, it must induce truthful reporting as a dominant strategy \citep{Bergemann-Morris-2005-Eca}.

In our principal--worker environment, the SCF satisfies this requirement. As shown in Table~\ref{Table:DM}, regardless of the other worker’s report, each worker is weakly better off reporting their true type. Furthermore, an expert is strictly better off reporting truthfully whenever the other worker reports being a beginner.

While our motivating example is simple, the limitations it reveals are not specific to this environment. To formalize the impossibility result illustrated by the example, we now embed it in a general private-type implementation framework.

Let the outcome space be $X$, the set of agents be $N=\{1,\ldots,n\}$, and the type space be $\Theta=\Theta_1 \times \cdots \times \Theta_n$. For each agent $i \in N$ and type $\theta_i \in \Theta_i$, preferences are represented by a payoff function $\pi_i(\cdot \mid \theta_i)$.

A mechanism is a pair $(M,\varphi)$, where $M = (M_i)_{i \in N}$ is the message space and $\varphi : M \rightarrow X$ is the outcome function. A strategy for agent $i$ is a function $\sigma_i : \Theta_i \rightarrow M_i$. An ex post equilibrium of $(M,\varphi)$ is a strategy profile $\sigma = (\sigma_i)_{i \in N}$ such that, for every $\theta \in \Theta$, every $i \in N$, and every $m_i \in M_i$,
\[
\pi_i\bigl(\varphi(\sigma(\theta)) \mid \theta_i\bigr)
\ge
\pi_i\bigl(\varphi(m_i,\sigma_{-i}(\theta_{-i})) \mid \theta_i\bigr).
\]

\begin{table}[t]
\centering
\begin{tabular}{c|cc}
 & \multicolumn{2}{c}{Worker~2 reports} \\
Worker~1 reports & Beginner & Expert \\ \hline
Beginner & $(L,M),(L,M)$ & $(L,P),(H,D)$ \\
Expert   & $(H,D),(L,P)$ & $(H,M),(H,M)$
\end{tabular}
\caption{Direct revelation mechanism. The matrix represents the contract given to each agent conditional on the reports. In each entry, the contract of worker~1 is shown on the left and that of worker~2 on the right. For example, $(L,M),(L,M)$ denotes the outcome in which both worker~1 and worker~2 receive the contract $(L,M)$.}
\label{Table:DM}
\end{table}

\subsection{A limit of standard robust implementation}

Ex post implementation ensures that the principal’s desired outcome arises in Nash equilibrium across a robust class of information models. However, mechanisms that achieve ex post implementation may also admit Nash equilibria that lead to outcomes different from those intended by the principal. Such undesirable equilibria can be not only robust, but also plausible. 

Our principal--worker environment illustrates this issue. The direct revelation contract in Table~\ref{Table:DM} admits a dominant-strategy equilibrium that yields outcomes not intended by the principal. In particular, it is a weakly dominant strategy for a beginner to report being an expert. As a result, the dominant-strategy equilibrium in which all workers report being experts Pareto dominates the truthful equilibrium from the workers’ perspective. If all workers are beginners, both are better off coordinating on reporting expertise. If there is only one expert, that worker also benefits from performing the mixed task rather than the entire delicate task. We refer to this equilibrium as \emph{worker-optimal}.

The limitation of the principal’s direct revelation mechanism cannot be overcome within the standard implementation framework. In fact, any mechanism that implements the principal’s objective in ex post equilibrium also implements the worker-optimal SCF in ex post equilibrium.

The following proposition formalizes this statement in a general private-type implementation environment. We credit this result to \citet{Repullo-1985-RES} who proved an essentially equivalent form for dominant-strategy equilibria.

\begin{proposition}[\citealp{Repullo-1985-RES}]\rm
Let $f:\Theta\rightarrow X$ be an scf; $\delta$ be an ex post equilibrium of the direct revelation mechanism of $f$, $(\Theta,f)$; and $g:=f\circ\delta$ the scf induced by this equilibrium.   Let $(M,\varphi)$ be a mechanism for which there is an ex-post equilibrium $\sigma$ such that for each $\theta\in\Theta$, $\varphi(\sigma(\theta))=f(\theta)$. Then, $\sigma\circ\delta$ is an ex post equilibrium of $(M,\varphi)$ whose outcome scf coincides with $g$.
\label{Prop:Repullo}
\end{proposition}

Proposition~\ref{Prop:Repullo} clarifies the limits of robust implementation. Let $f$ denote the principal’s optimal SCF and $g$ the worker-optimal SCF. Since $g$ arises as the outcome SCF of an ex post equilibrium of the direct revelation mechanism for $f$—indeed, a dominant-strategy equilibrium—any mechanism that robustly implements $f$ must also robustly implement $g$.

This implies that whenever the principal provides information-robust incentives for workers to non-cooperatively coordinate on the principal’s preferred outcome, workers also retain the ability to non-cooperatively coordinate on their own preferred outcome. If workers care only about their payoffs, it is plausible that they will coordinate on the Pareto-dominant outcome whenever it is supported by non-cooperative incentives. Consequently, to identify mechanisms that perform well in practice, the principal must appeal to regularities in workers’ behavior beyond payoff maximization.

\subsection{Behavioral robust implementation}\label{Sec-behav-robust}

Aversion to lying is a well-documented behavioral regularity.\footnote{For formal definitions distinguishing lies from deception and for an analysis of how they operate in strategic settings, see \citealp{Sobel-2020-JPE}.} A canonical experimental design asks subjects to roll a die in private and report the outcome, which determines their payoff—for example, \$50 if the outcome is six and zero otherwise \citep[e.g.,][]{Fischbacher-Follmi-JEEA-2008}. Although the experimenter cannot verify individual reports, the true distribution of outcomes is known, allowing one to detect whether a non-negligible fraction of subjects reports truthfully. 
A large meta-analysis of experiments of this type—covering ninety papers with approximately forty thousand subjects across forty-seven countries—shows that subjects forgo about three-quarters of their potential gains from lying \citep{Abeler-2019-ECA}. 

If workers are averse to lying, a principal can exploit this tendency to increase the probability of achieving their optimal outcome. Instead of using abstract message spaces, the principal can ask questions with concrete true values to workers. 

A first approach is obvious. Recall that the principal's SCF is strategy-proof. Thus, a beginner's misreport does not have a direct non-cooperative incentive. A lying-averse worker, whose aversion to lie is enough to overcome their drive to cooperate with other workers, will choose to report truthfully in the principal's direct revelation mechanism.  Thus, if the principal is to use this mechanism, they should do it in its explicit formulation, in which agents are asked directly for their types. 

To avoid ambiguities, we refer to this mechanism as the  \emph{explicit direct mechanism} (\DE ). We refer to a strategically equivalent mechanism in which actions have neutral labels---so no message is a lie for any agent---as an \emph{implicit direct mechanism} (\DI ). Since \DE\ and \DI\ are isomorphic, we can use the same language to refer to the actions and strategies in both games. We refer in \DI\ to the \textit{truthful action} for an agent, as that corresponding to reporting their true type in \DE; and  we refer in \DI\ to the \textit{deceptive action} for an agent, as that corresponding to reporting their opposite type in \DI. By using \DE\ instead of some \DI\ form, the principal can exploit the lying aversion among beginner types to counteract their cooperative behavior.

The principal’s explicit direct revelation mechanism is not the only way to exploit lying aversion. In \DE, misreporting remains a dominant action for beginners. As a result, such misreports may still serve as a focal point for workers who are not averse to lying. The design challenge is therefore to reduce the attractiveness of misreporting for these workers without undermining the appeal of truthful reporting for lying-averse agents.

The principal can address this challenge by enlarging the message space of the direct revelation mechanism and using the additional messages to weaken the incentive to misreport. 

The following mechanism, which we call the \emph{explicit extended mechanism} (\IE), achieves this goal. Each worker is asked to report their expertise level, with the additional option of leaving the question unanswered. The principal pays a high salary to reported experts and a low salary to reported beginners, splits tasks among workers of equal reported expertise, and assigns the delicate task to experts whenever at least one is reported. If exactly one worker leaves the question unanswered, the principal assigns that worker an expertise level opposite to that of the reporting worker. If both workers leave the question unanswered, both are treated as beginners. Table~\ref{Table:DM5} summarizes the resulting mechanism.

\begin{table}[ht]
\centering
\begin{tabular}{c|ccc}
 & \multicolumn{3}{c}{Worker~2 reports} \\
Worker~1 reports & B & E & U \\ \hline
B & $(L,M),(L,M)$ & $(L,P),(H,D)$ & $(L,P),(H,D)$ \\
E & $(H,D),(L,P)$ & $(H,M),(H,M)$ & $(H,D),(L,P)$ \\
U & $(H,D),(L,P)$ & $(L,P),(H,D)$ & $(L,M),(L,M)$ \\
\end{tabular}
\caption{\IE\ mechanism. Each worker reports beginner $(B)$, expert $(E)$, or leaves the question unanswered $(U)$. The matrix shows the allocation conditional on the reports. In each cell, worker~1’s outcome appears on the left and worker~2’s on the right. For example, $(L,M),(L,M)$ assigns $(L,M)$ to both workers.}
\label{Table:DM5}
\end{table}

It is a dominant strategy for each worker to report their true type in \IE. The mechanism retains the non-pecuniary incentives for lying-averse agents: under any common prior, in any equilibrium in which lying-averse agents would lie only if doing so yields a higher payoff, they report truthfully. By Proposition~\ref{Prop:Repullo}, \IE\ also admits an ex post equilibrium that delivers the workers’ optimal outcome (all agents report $E$). However, unlike in both direct mechanisms, this equilibrium is not in dominant strategies. Thus, when play is noisy, workers bear some risk when attempting to coordinate on their optimal outcome.

We refer to a neutral form of \IE\ as the \emph{implicit extended mechanism} (\II): a mechanism strategically equivalent to \IE\ in which actions have neutral labels, so no message is a lie for any agent. Since \IE\ and \II\ are isomorphic, we use the same language for actions and strategies in both games. In \II, the action corresponding to revealing one’s true type in \IE\ is the \emph{truthful action}; the action corresponding to reporting the opposite type is the \emph{deceptive action}; and the remaining action is \emph{unanswered}, mirroring the unanswered option in \IE. 

\section{The experiment}\label{Sec-exp-design}
\subsection{Experimental design}
We implemented the cooperative Bayesian game described in Section~\ref{Section-problem} under a symmetric common prior. We utilized a $2\times 2$ between-subject design, varying whether the mechanism was direct and whether messages were explicit, \{Direct, Extended\}$\times$\{Implicit, Explicit\}. We refer to the four possible mechanisms as \DI, \DE, \II, and \IE, respectively.

Experimental sessions had either nine or twelve subjects. At the start, $2/3$ of subjects were randomly assigned the role of \emph{worker} and $1/3$ the role of \emph{mechanism designer} (the “staffer”). Workers were randomly assigned types (“expert” or “beginner”) with equal probability; types were private information, while the assignment probabilities were common knowledge. Roles were fixed across the experiment, but worker types were re-drawn each period, so a worker subject could be a beginner in some periods and an expert in others. Staffers remained staffers throughout.

Two workers and one staffer were randomly matched each period. The two workers simultaneously chose messages; the mechanism then determined allocations based on the message profile. The staffer had no strategic moves but received payoffs that depended on worker types and messages.

Worker payoffs matched \citeauthor{Repullo-1985-RES}’s original example \cite[pp.~224--225]{Repullo-1985-RES}: for experts,
$\pi_i(H,M\mid E)=4$, $\pi_i(H,D\mid E)=2$, $\pi_i(L,M\mid E)=1$, $\pi_i(L,P\mid E)=0$;
for beginners,
$\pi_i(L,P\mid B)=4$, $\pi_i(H,M\mid B)=4$, $\pi_i(L,M\mid B)=2$, $\pi_i(H,D\mid B)=2$.
These values satisfy \eqref{relexp} and \eqref{relbeg}. The mechanism designer—absent from the original example—received 5 if the mechanism correctly identified both workers’ types, 3 if it identified exactly one, and 1 if it identified neither. The purpose of this third player was to make misrepresentation meaningful and harmful to another party (i.e., not a “white lie”) and hold constant the total payoffs of the three players across equilibria.

The four mechanisms followed the game tables in Section~\ref{sec:RepulloGame}. Specifically, the $2\times 2$ Table~\ref{Table:DM} is the game played between workers in \DE, and the $3\times 3$ Table~\ref{Table:DM5} is the game played in \IE. \DI\ and \II\ used neutral-label versions of \DE\ and \IE, respectively. In the explicit mechanisms, strategies were labeled “Beginner,” “Expert,” and (when applicable) “Decline to State,” and subjects had to type the full statement “I am a beginner,” “I am an expert,” or “I decline to state” to implement their choice. With the implicit mechanisms, the corresponding strategies were labeled “Option~A,” “Option~B,” and (when applicable) “Option~C,” and subjects likewise had to type the full label.

The choice screen displayed three payoff tables—one each for the subject, the other worker, and the staffer—covering all pure-strategy profiles for each of the four worker-type pairings. Subjects used toggle buttons to switch across the four type pairs. A fixed message at the top of the screen reminded each worker of their current type. Staffers saw a similar screen but did not make a choice. Figures~\ref{fig:screen}(a) and \ref{fig:screen}(b) show screenshots for an expert worker under the explicit extended (\IE ) and implicit extended mechanisms (\II ), respectively. Direct mechanisms used identical screens except that the leave unanswered option was absent from the tables and not selectable. The staffer screen mirrored the information displays but never allowed a choice.

\begin{figure}
\includegraphics[width=\textwidth]{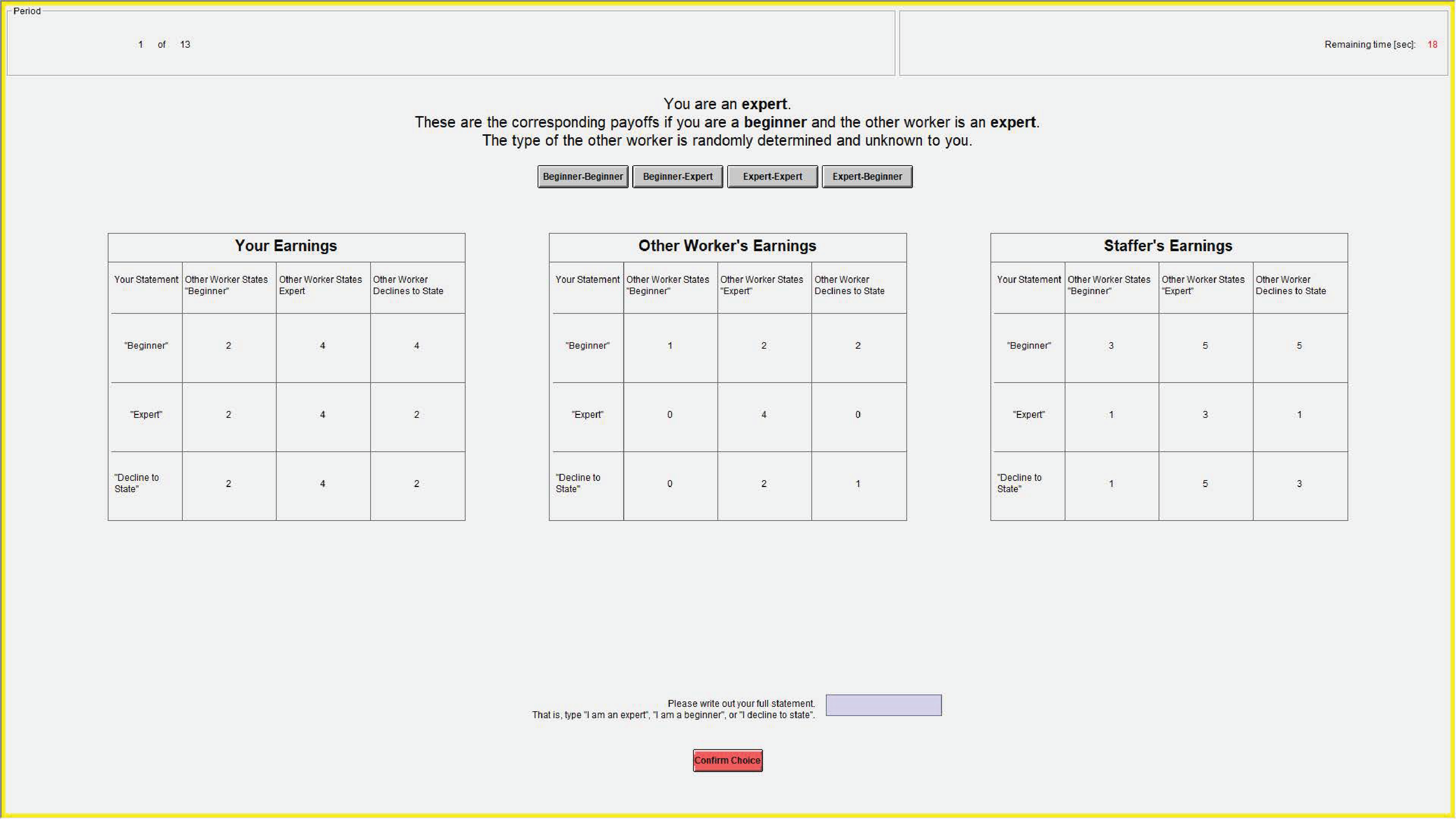}\\

\includegraphics[width=\textwidth]{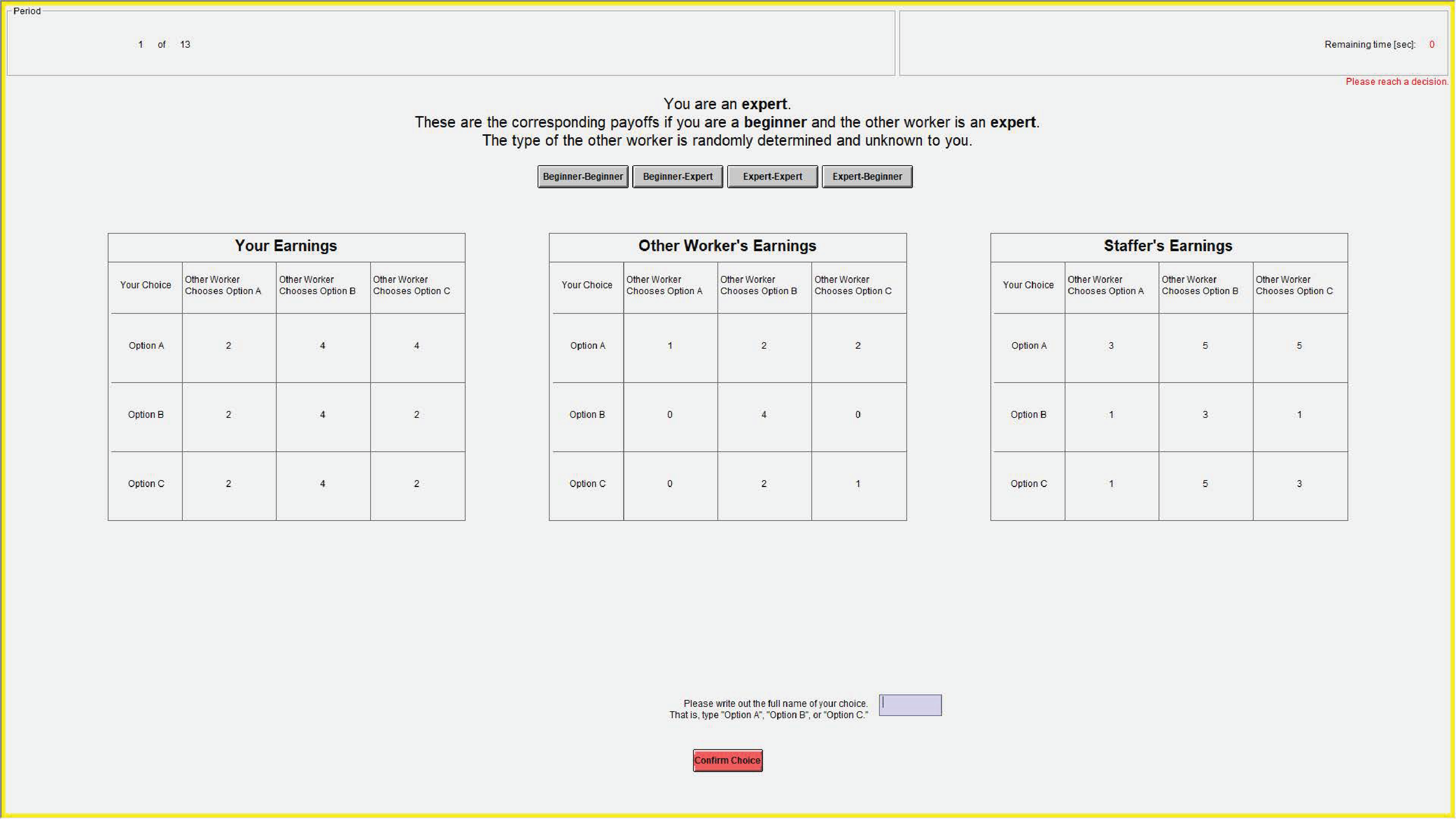}
\caption{\label{fig:screen} Screens for the \IE\ (a, top) and \II\ (b, bottom) mechanisms. The direct mechanisms \DI\ and \DE\ 
used identical screens except there were no rows and columns in the table for the unanswered option, nor option for subjects to select.}
\end{figure}

\subsection{Experimental procedures}

Fourteen sessions with a total of 159 subjects were conducted at the Economic Research Laboratory, Texas A\&M University, during April--May 2021. For \DI\ and \DE, three 12-subject sessions were run for each mechanism. For \II, four sessions were run (three with 12 subjects and one with 9). For \IE, four sessions were run (two with 12 subjects and two with 9).

Instructions were read aloud and subjects were required to compete a comprehension quiz before participating in the experiment (see Supplemental Appendix for both). Each session comprised 13 periods. At the end of each period, subjects observed the actions and payoffs of all group members. The first three periods were unpaid practice. Following \citet{azrieli2018incentives} and \citet{brown2018separated}, one of the last ten periods was randomly selected for payment. Earnings were converted at a rate of $1 \text{ECU} = \$1$, plus a \$10 participation payment. Average earnings were \$13.61 for a 40-minute session.

Subjects were recruited using ORSEE \citep{greiner} and made decisions using z-Tree \citep{fischbacher}. In accordance with university regulations at the time, all subjects wore masks and maintained a social distance of six feet.

\subsection{Hypotheses}

Across all four mechanisms, we compare the incidence of the worker-optimal equilibrium, a meaningful departure from the designer’s intended equilibrium. Although theory specifies strategies ex ante, we observe actions conditional on type. When both workers are experts, both equilibria predict the same action pair (both truthfully declare expert), and each player’s action is uniquely dominant. Under all mechanisms, we therefore expect a high and similar frequency of this action pair when both workers are experts.

We instead focus on cases with at least one beginner: (i) one beginner and one expert, and (ii) two beginners. The worker-optimal and truthful equilibria predict that beginners play the deceptive and truthful actions, respectively. In the \DI\ baseline, both actions are weakly dominant, and the worker-optimal equilibrium improves payoffs for both workers. We thus expect to observe action pairs resembling the worker-optimal equilibrium with some frequency, and we treat this rate as a baseline for cross-mechanism comparisons.

Although \DE\ is strategically identical to \DI, under \DE\ a beginner must enter a statement known to be objectively false to play the deceptive action. Empirical evidence shows subjects are averse to making such lies and act as if they carry an additional (psychic) cost. Therefore, we predict that beginners choose the deceptive action less often under explicit mechanisms, and action pairs associated with the worker-optimal equilibrium occur less frequently relative to \DI.

\begin{hypothesis}\label{hyp:lie}
The worker-optimal equilibrium is observed less often when beginner types must lie explicitly to play it.
\end{hypothesis}

By adding an additional irrelevant action, the $3\times 3$ extended mechanisms, \II\ and \IE, remove weak dominance from the deceptive action for beginners. Further, these game forms leave the truthful action as the unique weakly dominated action for both types. The truthful equilibrium is the only equilibrium that satisfies dominance. We therefore expect more action pairs associated with the truthful action under \II\ and \IE.

\begin{hypothesis}\label{hyp:dom}
The worker-optimal equilibrium is observed less often when it is not in dominant strategies.
\end{hypothesis}

\section{Results}\label{Sec-Results}

Our two focal equilibria, the truthful and the worker-optimal, coincide in expert behavior but differ for beginners. Under the truthful equilibrium, types always report truthfully, so the action pair (expert, expert) should be observed only when both workers are in fact experts. In contrast, under the worker-optimal equilibrium, (expert, expert) should be observed regardless of the workers’ types. Figure~\ref{fig:bars}(a) reports the frequency of (expert, expert) by type-pair across the four mechanisms; Figure~\ref{fig:bars}(b) reports the frequency with which both workers report their true type. Note that for two expert types these outcomes coincide.

With two expert types, we observe high rates (90--100\%) of the (expert, expert) action pair across all mechanisms. For an expert, reporting expert is the uniquely strictly dominant action consistent with both equilibria. By contrast, when at least one worker is a beginner (one beginner--one expert or two beginners), rates differ across mechanisms. In the baseline \DI, (expert, expert) is observed a majority of the time under all type combinations; overall, this worker-optimal action pair occurs in 60.8\% of cases. Meanwhile, both types report truthfully only 28.3\% of the time, almost entirely when both are experts. Under the other three mechanisms—most notably the explicit ones—we observe lower rates of the (expert, expert) action pair and higher rates of truthful reporting when at least one worker is a beginner.

\begin{figure}[t]\begin{center}
\includegraphics[width=0.49\textwidth]{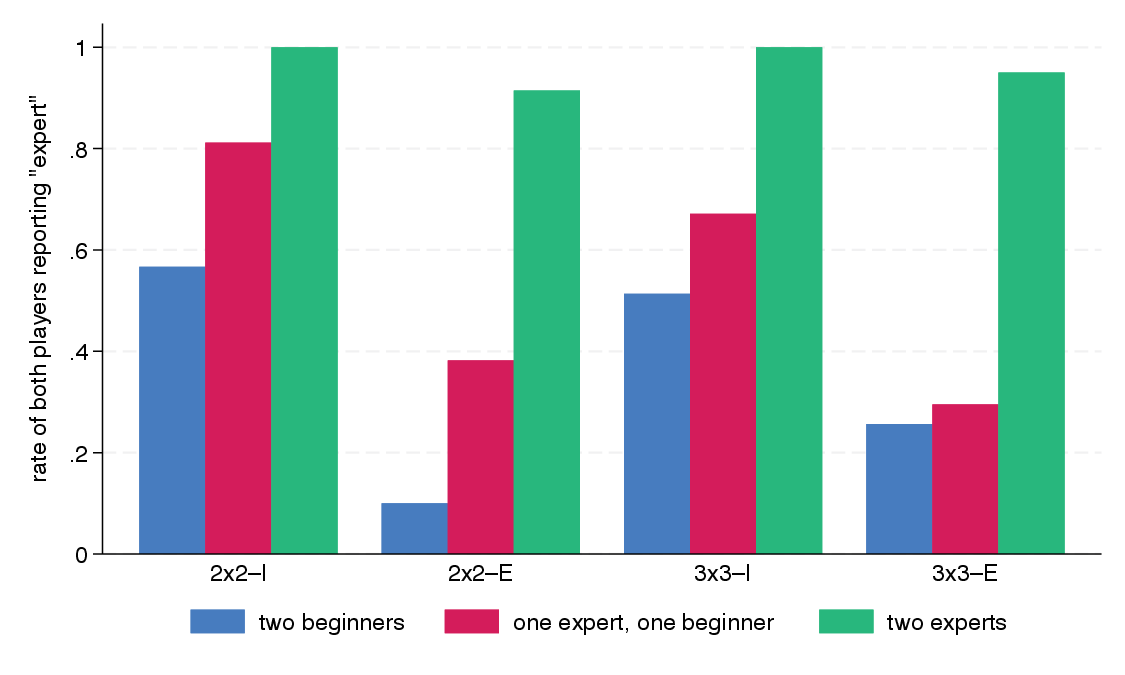}
\includegraphics[width=0.49\textwidth]{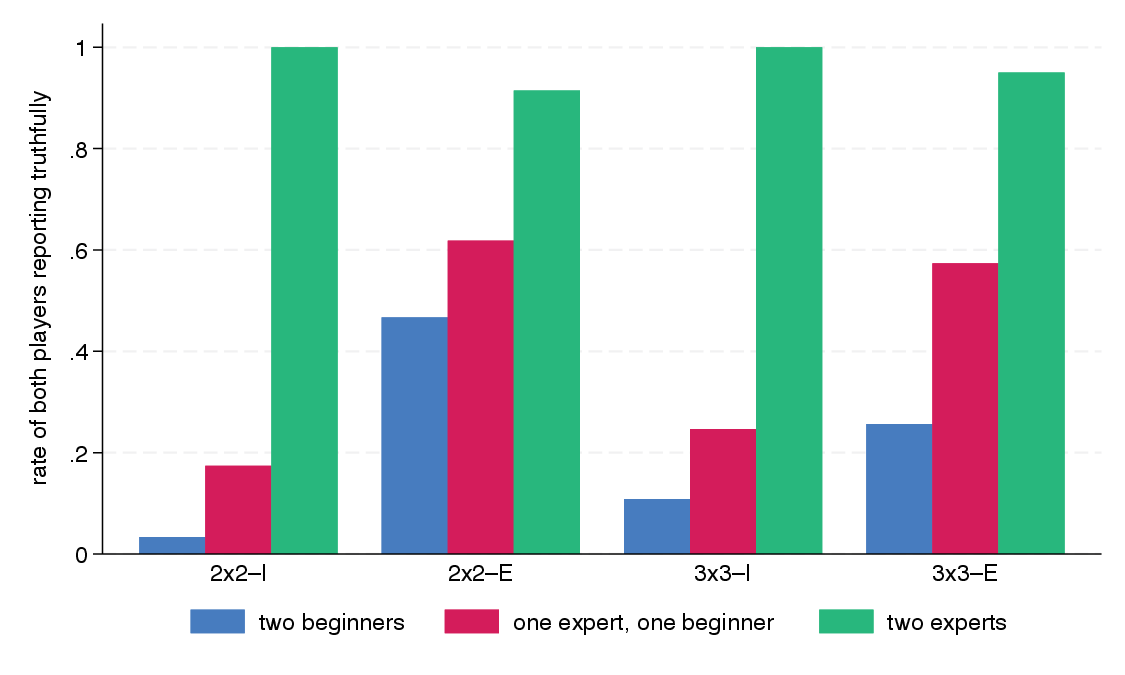}\\
\includegraphics[width=0.49\textwidth]{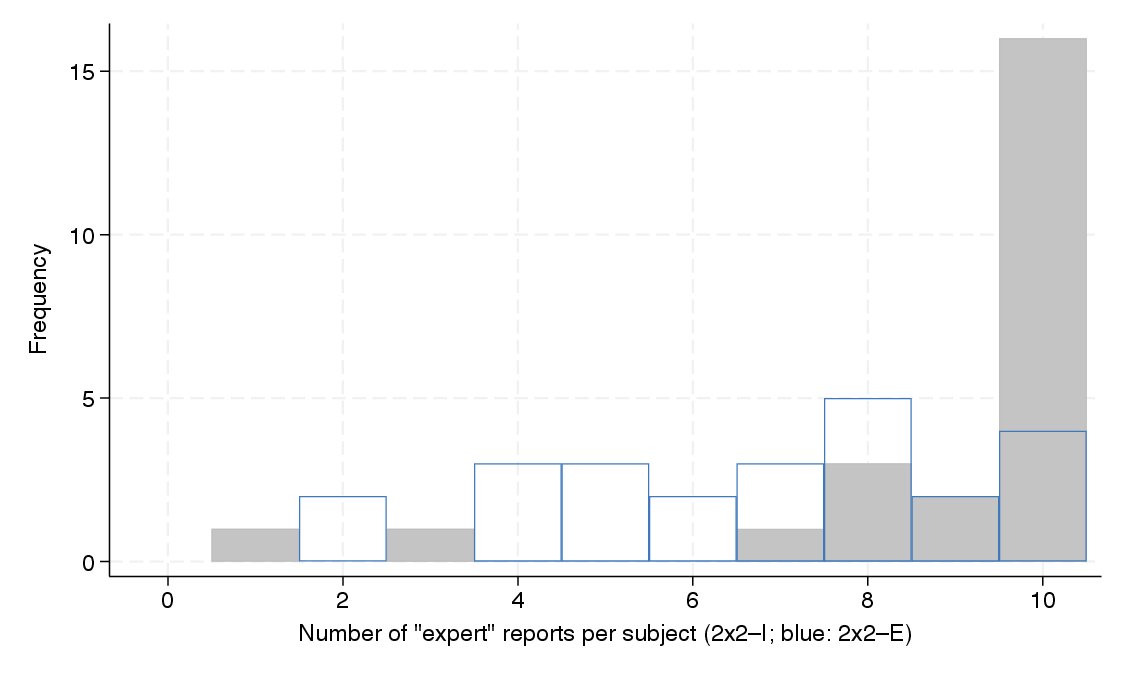}
\includegraphics[width=0.49\textwidth]{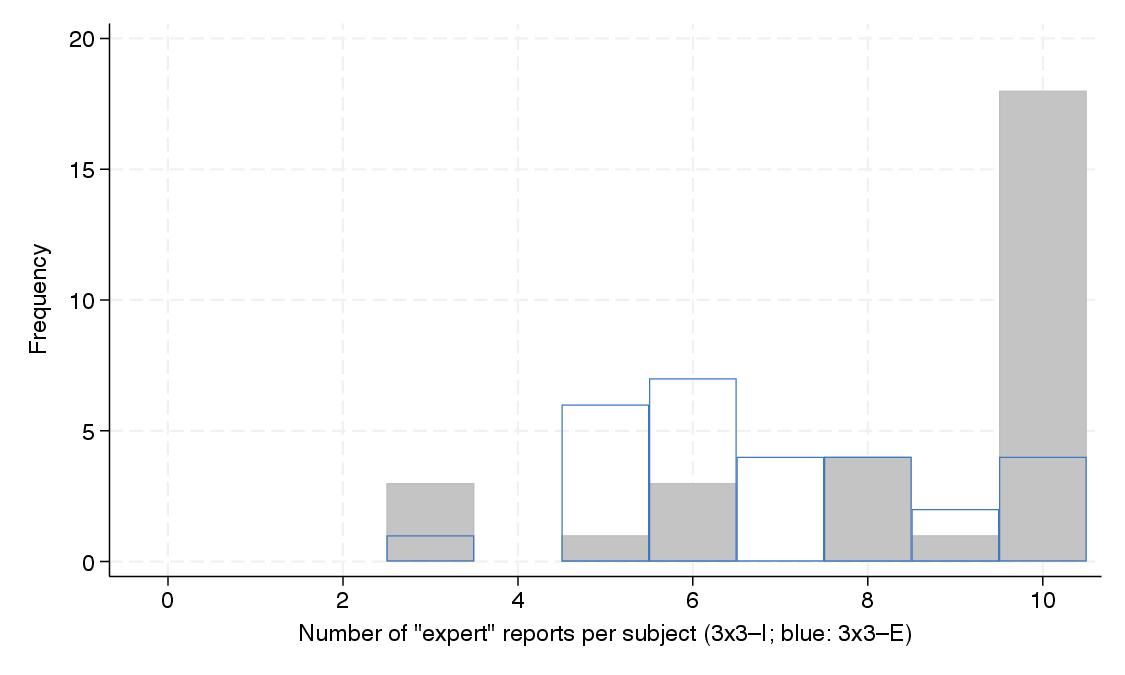}
\caption{(a, upper left) Rates of both workers reporting export by mechanism and player type-pair. (b, upper right) Rates of both players reporting type truthfully by mechanism and player type-pair. (c, bottom left) Histogram of the number of times (out of 10) players under the two direct mechanisms (\DE\ and \DI ) report being an expert type ($N=24$ and $30$, respectively). (d, bottom right) Histogram of the number of times (out of 10) players in the extended mechanisms (\IE\ and \II ) report being an expert type ($N=24$ and $30$, respectively).\label{fig:bars}}
\end{center}
\end{figure}

\begin{figure}[p]\begin{center}
\includegraphics[width=0.625\textwidth]{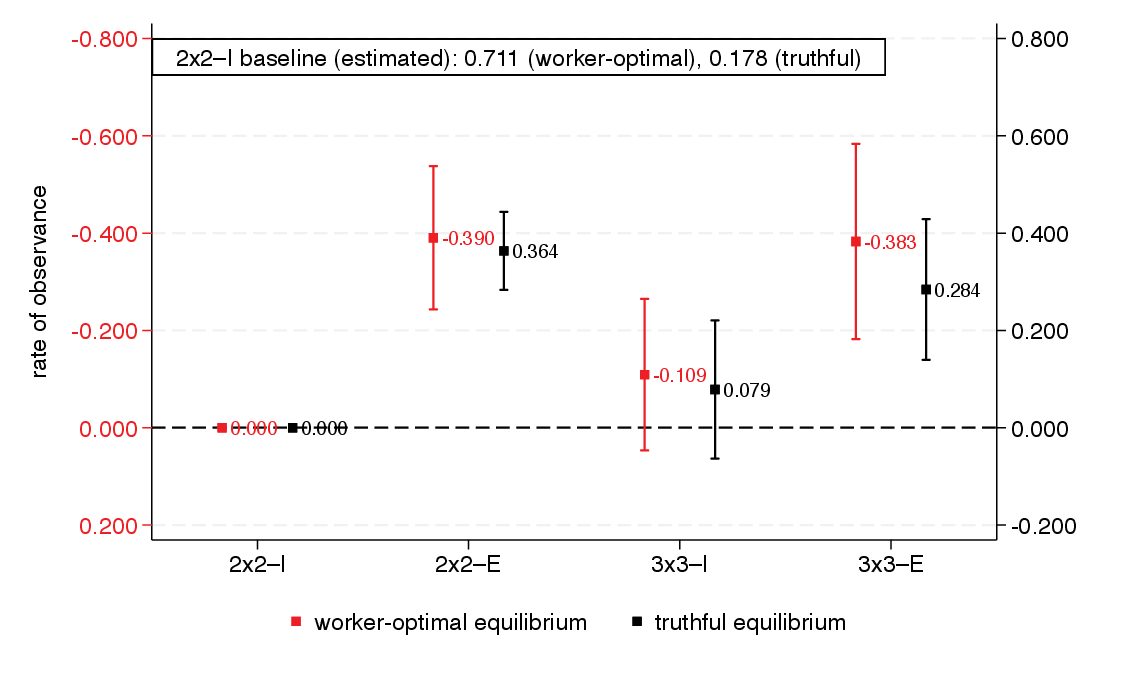}
\includegraphics[width=0.625\textwidth]{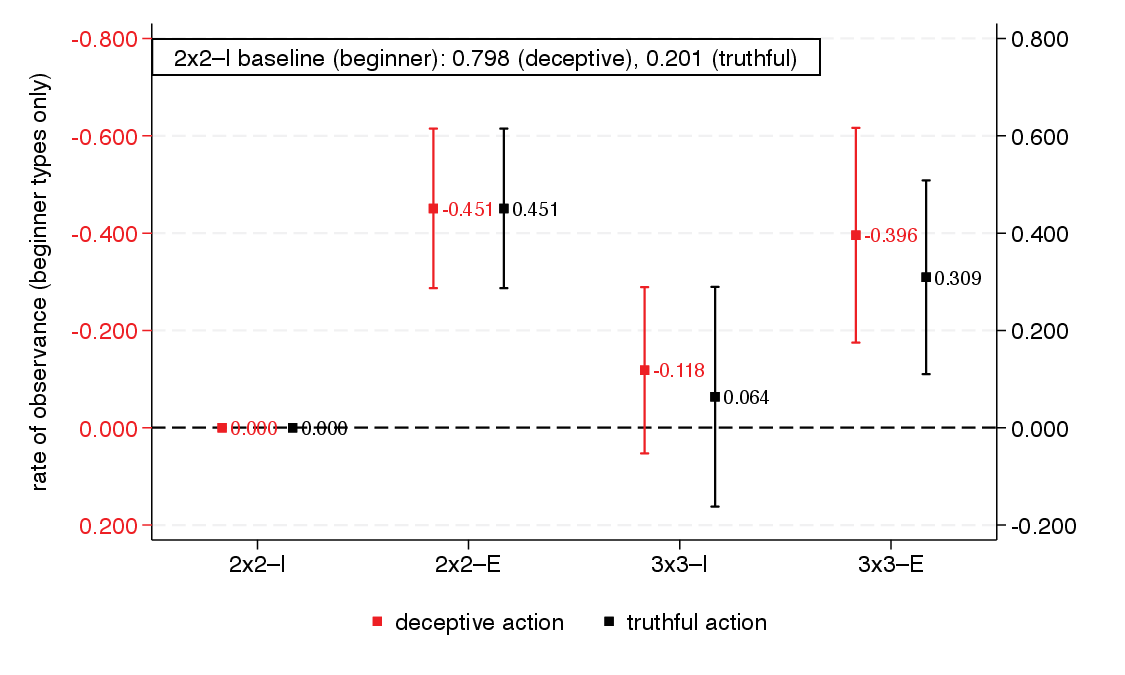}
\includegraphics[width=0.625\textwidth]{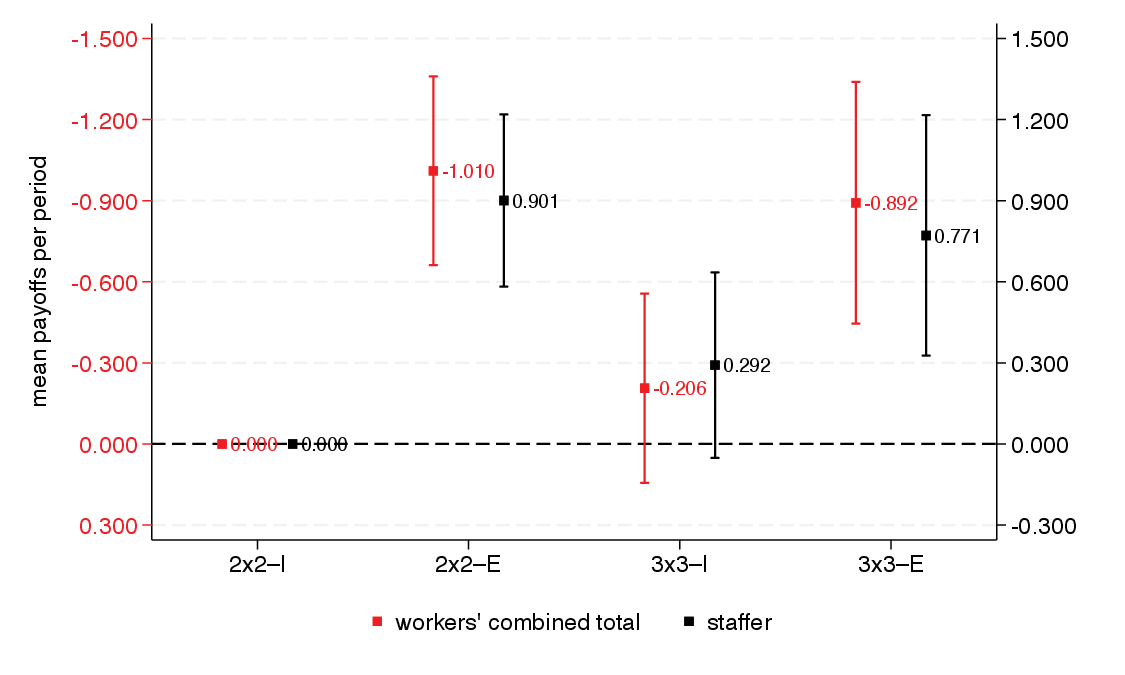}
\caption{Each figure reports the estimated coefficients of two regression models. All are relative to the baseline \DI\ mechanism. The plots in black (right bracket intervals whose axis is labeled on the right) show positive effects; the plots in red (left intervals whose axis is labeled on the left in reversed values) show negative effects. (a, top) Effect of each mechanism on observance of truthful equilibrium (in black) and worker-optimal equilibrium (in red). (b, top) Effect of each mechanism on play of truthful action (in black) and deceptive action (in red) when subjects are beginner type. (c, bottom) Effect of each mechanism on mean worker combined (in red) and staffer (in black) payoffs per period. All regressions utilize cluster robust standard errors at the session level. The regressions in the top and bottom panels include dummy variables for worker-type pair (i.e., beginner-beginner, beginner-expert, expert-expert). Results in tabular form in Supplemental Appendix Table \ref{tab:regs}.\label{fig:plusminus}}
\end{center}
\end{figure}

We estimate how each mechanism affects the incidence of each equilibrium using two linear probability models (LPMs). The dependent variable is an indicator for the worker-optimal equilibrium (Model~1) or the truthful equilibrium (Model~2). Regressors are dummies for mechanisms, with \DI\ as the baseline. Because when both workers are experts the observed action pair (expert, expert) is consistent with both equilibria, we treat these cases as censored. We include dummies for the three type-pair permutations and cluster robust standard errors at the session level. Figure~\ref{fig:plusminus}(a) reports estimates: red bars for the worker-optimal equilibrium and black bars for the truthful equilibrium.

Relative to \DI, both explicit mechanisms reduce the estimated likelihood of the worker-optimal equilibrium by 38--39 percentage points (pp, $p<0.001$). Simultaneously, \DE\ increases the truthful equilibrium by 36~pp ($p<0.001$) and \IE\ increases by 28~pp ($p<0.001$).
For \II, which leaves the worker-optimal equilibrium only in dominated strategies, the effects are smaller: $-11$~pp (worker-optimal) and $+8$~pp (truthful), neither statistically different from zero ($p\approx0.170$ and $p\approx0.278$, respectively).

We also examine subjects’ choice of strategy at the individual level. Each worker faced ten equiprobable draws of type and made ten reports. Figures~\ref{fig:bars}(c)--(d) plot the total number of “expert” reports, by subject, for the direct and extended mechanisms, respectively. Under truthful play, the distribution would be binomial with mean~5; under play consistent with the worker-optimal equilibrium, all values would be 10. All mechanisms show a right shift, indicating some strategic misreporting. Pronounced spikes at~10 appear under \DI\ and \II: 16/24 (67\%) and 17/30 (57\%) subjects, respectively, play in line with the worker-optimal equilibrium. Under \DE\ and \IE, the spikes are much smaller (4 subjects each: 17\% and 14\%).

Figure~\ref{fig:plusminus}(b) reports LPM results for (i) playing the deceptive action and (ii) playing the truthful action, restricting to observations in which the subject is a beginner. For the extended mechanisms, choosing the unanswered option is coded as neither. As predicted, beginners are substantially less likely to choose the deceptive action when doing so requires an explicit lie: \DE\ reduces “expert” claims by 45~pp (and, by construction, identically raises truthful reports by 45~pp, $p<0.001$), and \IE\ reduces deceptive play by 40~pp ($p<0.001$). Under \II, the reduction is 12~pp and not statistically significant $(p\approx0.159)$. With the two extended mechanisms, the increase in truthful play is marginally smaller than the decrease in lying, consistent with some substitution into the unanswered option.

Mechanism differences meaningfully affect designer payoffs. Recall the staffer earns 5 if both types are correctly identified, 3 if exactly one is, and 1 otherwise. Under \DI, the staffer’s mean payoff is \$3.27 and workers’ total is \$7.52; the combined \$10.78 is near the efficient maximum of \$11. Under \DE, higher truthful play raises staffer payoffs by \$0.90 and reduces worker totals by \$1.01 ($p<0.001$, both estimates). Under \IE, the staffer’s payoff rises by \$0.77 and workers’ falls by \$0.89 ($p<0.001$, both estimates). The \II\ increases staffer payoffs by 0.29, a result that is marginally significant ($p\approx0.089$). The decrease to worker payoffs is smaller, 0.20, and not statistically meaningful ($p\approx0.201$).

Taken together, the results strongly support Hypothesis~\ref{hyp:lie}. When equilibrium play requires subjects to lie explicitly, the incidence of the worker-optimal equilibrium falls substantially, and behavior shifts toward the truthful equilibrium. By contrast, we find less support for Hypothesis~\ref{hyp:dom}: making the worker-optimal equilibrium consist only of weakly dominated strategies yields only a modest reduction in its incidence. Crucially, across all six regressions, the effects of mechanisms that leverage lying aversion are significantly larger than those that rely on unique weak dominance. Specifically the magnitude of the coefficients for the \DE\ and \IE\ are larger than \II\ ($p<0.05$) in all 12 binary comparisons.

We caution that the small effects of the \II\ mechanism should not necessarily be attributed to irrationality. For one, it is perfectly rational to play a weakly dominated action as a best response in equilibrium. Further, in the extended mechanisms, the deceptive action is not especially costly for beginners: payoffs diverge from the truthful action only when the other worker plays ``unanswered,'' an off-equilibrium action. Such choice occurs with frequency $2.7\%$ and $4.6\%$ in the \II\ and \IE\ mechanisms, respectively. In \IE, this implies an expected loss relative to the dominant action of $(4-2)\times 0.027 = 0.054$, about $1.5\%$ of the \$3.69 average earnings. It is slightly larger in \IE, \$0.092 (9.2 cents), or $2.7\%$ of the \$3.35 average earnings. Neither magnitude approaches the gains from worker cooperation.

\section{Discussion and concluding remarks}\label{Sec-Discussion}

This paper evaluates, empirically, whether lying aversion can be used to improve mechanism performance in environments where agents have incentives to coordinate on outcomes that the designer does not intend. Theoretically, we observe that robust implementation alone cannot preclude worker-optimal equilibria: any mechanism that robustly implements the principal’s target also robustly implements the worker-optimal outcome (Proposition~\ref{Prop:Repullo}). Experimentally, we find that making a deceptive action an \emph{explicit lie} is a powerful and practical tool for shifting play away from such outcomes and toward truth-telling.

Three conclusions emerge. First, across mechanisms, explicit language that turns deceptive actions into a lie substantially reduces play of the worker-optimal equilibrium and increases truthful outcomes; the estimated effects are large and statistically robust. Second, relying only on dominance (via an added ``unanswered'' option) yields smaller and statistically weaker effects. In our data, lying aversion is a stronger deterrent than weak dominance when payoffs are similar across actions. Third, these behavioral shifts translate into economically meaningful improvements in designer payoffs without large efficiency losses.



Taken together, our results suggest that lying aversion is not merely an ancillary psychological detail but a practically useful force for mechanism designers. In environments where robust implementation leaves room for undesirable equilibria, carefully crafted message frames that make false claims costly in non-pecuniary domains can help align behavior with social objectives.

\section*{Appendix}

\begin{proof}[Proof of Proposition~\ref{Prop:Repullo}]
Since $\delta$ is an ex-post equilibrium of $(\Theta,f)$, we have that for each $\theta\in\Theta$ and each $i\in N$, $\pi_i(f(\delta(\theta))|\theta_i)\geq \pi_i(f(\theta_i,\delta_{-i}(\theta_{-i}'))|\theta_i)$. Since $\varphi\circ \sigma=f$,
\[\pi_i(\varphi_i(\sigma_i(\delta_i(\theta_i)),\sigma_{-i}(\delta_{-i}(\theta_{-i})))|\theta_i)\geq \pi_i(\varphi_i(\sigma_i(\theta_i),\sigma_{-i}(\delta_{-i}(\theta_{-i})))|\theta_i).\]
Since $\sigma$ is an ex-post equilibrium of $(M,\varphi)$, for each $m_i\in M_i$,
\[ \pi_i(\varphi_i(\sigma_i(\theta_i),\sigma_{-i}(\delta_{-i}(\theta_{-i})))|\theta_i)\geq \pi_i(\varphi_i(m_i,\sigma_{-i}(\delta_{-i}(\theta_{-i})))|\theta_i).\]
Then, for each for each $\theta\in\Theta$, each $i\in N$, and each $m_i\in M_i$,
\[\pi_i(\varphi_i(\sigma\circ \delta(\theta))|\theta_i)\geq \pi_i(\varphi_i(m_i,\sigma_{-i}\circ \delta_{-i}(\theta_{-i}))|\theta_i).\]
Thus, $\sigma\circ\delta$ is an ex-post equilibrium of $(M,\varphi)$ whose outcome scf is $\varphi\circ \sigma\circ\delta=f\circ \delta=g$.
\end{proof}

\bibliography{ref-TMD}
\pagebreak
\appendix

\renewcommand{\thetable}{A.\arabic{table}}
\renewcommand\thefigure{A.\arabic{figure}}    
\renewcommand\theequation{A.\arabic{equation}}    
\setcounter{equation}{0}
\setcounter{table}{0}
\setcounter{figure}{0}

\centerline{\large \textbf Supplemental Appendix: Not Intended for Publication}
\section{Additional Tables}

\begin{table}[ht]
\resizebox{\textwidth}{!}{
\begin{tabular}{lcccccc}
\toprule
\midrule
 & (1) & (2) & (3) & (4) & (5) & (6) \\
VARIABLES & \begin{tabular}{c}worker-optimal\\
 equilibrium\\
 observed \end{tabular} & \begin{tabular}{c}truthful\\
 equilibrium\\
 observed \end{tabular}& \begin{tabular}{c} deceptive\\ action\\ played \end{tabular} & \begin{tabular}{c} truthful\\ action\\ played \end{tabular} & \begin{tabular}{c}staffer\\profit\end{tabular} & \begin{tabular}{c} workers'\\combined\\profit\end{tabular}\\ 
 &  &  &  &  &  &  \\\midrule
\DE\  & -0.390*** & 0.364*** & -0.451*** & 0.451*** & 0.901*** & -1.010*** \\
mechanism & (0.075) & (0.041) & (0.076) & (0.076) & (0.147) & (0.162) \\
\II\  & -0.109 & 0.079 & -0.118 & 0.064 & 0.292* & -0.206 \\
mechanism & (0.079) & (0.072) & (0.079) & (0.105) & (0.159) & (0.162) \\
\IE\  & -0.383*** & 0.284*** & -0.396*** & 0.309*** & 0.771*** & -0.892*** \\
mechanism & (0.102) & (0.074) & (0.102) & (0.092) & (0.206) & (0.207) \\
1 expert, & 0.177*** & 0.187*** &  &  & 1.181*** & 0.729*** \\
1 beginner & (0.056) & (0.055) &  &  & (0.177) & (0.185) \\
2 experts & 0.133* & 0.236*** &  &  & 2.175*** & 1.534*** \\
 & (0.072) & (0.066) &  &  & (0.266) & (0.242) \\
Constant & 0.586*** & 0.030 & 0.798*** & 0.202*** & 2.207*** & 6.829*** \\
 & (0.093) & (0.045) & (0.063) & (0.063) & (0.218) & (0.225) \\
 & (0.032) & (0.043) & (0.063) & (0.063) & (0.218) & (0.225) \\
 &  &  &  &  &  &  \\
Observations & 530 & 530 & 530 & 530 & 530 & 530 \\
R-squared &  &  & 0.139 & 0.132 & 0.434 & 0.259 \\
log likelihood & -290.0 & -267.1 & -340.7 & -335.9 & -748.0 & -825.2 \\ 
\midrule
\bottomrule
\multicolumn{7}{c}{ *** p$<$0.01, ** p$<$0.05, * p$<$0.1} \\
\end{tabular}}
\caption{Regressions found in Figure \ref{fig:plusminus}. Columns (1) and (2) show the regressions used in panel (a). Columns (3) and (4) show the regressions used in panel (b). Columns (5) and (6) show the regressions used in panel (c). All regressions
utilize cluster-robust standard errors at the session level.\label{tab:regs}}
\end{table}

\begin{table}[H]\footnotesize
\centering
\begin{threeparttable}
\begin{tabular}{lccc}
\toprule
\midrule
&\multicolumn{2}{c}{player type}&\\
\cmidrule{2-3}
mechanism&beginner only&expert only&all types\\
\midrule
\DI, 3 sessions &0.202&0.991&0.567\\
\DE, 3 sessions &0.652&0.976&0.821\\
\II, 4 sessions &0.265 (0.054)&1.000 (0.000) &0.640 (0.027)\\
\IE, 4 sessions &0.511 (0.086)&0.972 (0.007)&0.743 (0.046)\\
\midrule
$H_0:$ all mechanisms are equal\tnote{1}& 0.020 (0.001) & 0.073 (0.204) & 0.042 (0.029)\\
$H_0:$ direct and extended are equal\tnote{2} & 0.779 (0.976)& 0.965 (0.610)& 0.976 (0.662)\\
$H_0:$ implicit and explicit are equal\tnote{2}&0.001 (0.001) & 0.017 (0.070) & 0.006 (0.004)\\
\midrule
\bottomrule
\end{tabular}
\begin{tablenotes}
\item [1] Kruskal-Wallis p-value on session-level averages ($N=14$). Parentheses indicate results of test when third neutral actions are not considered in averages.
\item [2] Mann-Whitney p-value on session-level averages ($N=14$). Parentheses indicate results of test when third neutral actions are removed are not considered in averages.
\end{tablenotes}
\caption{Rates of play of truthful action and leaving unanswered (in parenthesis, where applicable) by mechanism and player type. Rates of deceptive action can be inferred by subtracting from 1.\label{tab:domrate}}
\end{threeparttable}
\end{table}
\end{document}